\pdfoutput=1
\documentclass[11pt]{article}
\usepackage{fullpage}
\usepackage{microtype}
\usepackage{mathtools,amsmath} \usepackage{amssymb,amsthm}
\usepackage{graphicx}
\usepackage{hyperref,color}

\newcommand{\sinn}{\ensuremath{S_{\textnormal{in}}}}
\newcommand{\sout}{\ensuremath{S_{\textnormal{out}}}}

\newcommand{\nn}[2]{\ensuremath{\textsc{n}_{#1}(#2)}}
\newcommand{\direct}[2]{{(#1,#2)}}

\newbox\ProofSym
\setbox\ProofSym=\hbox{%
	\unitlength=0.18ex%
	\begin{picture}(10,10)
	\put(0,0){\framebox(9,9){}}
	\put(0,3){\framebox(6,6){}}
	\end{picture}}

\graphicspath{{figures/}}

\newtheorem{theorem}{Theorem} 
\newtheorem{lemma}[theorem]{Lemma}

\begin{document}
\title{Reachability Problems for Transmission Graphs}
%
%

\author{Shinwoo An\thanks{Pohang University of Science and Technology,
		Korea. Email: {\tt{\{shinwooan, eunjin.oh\}@postech.ac.kr}}} 
	\and Eunjin Oh\footnotemark[1]}

\maketitle              
\begin{abstract}
	Let $P$ be a set of $n$ points in the plane where each point $p$ of $P$ is associated
	with a radius $r_p>0$. 
	The transmission graph $G=(P,E)$ of $P$ is defined as the directed graph 
	such that $E$ contains an edge from $p$ to $q$ if and only if $|pq|\leq r_p$
	for any two points $p$ and $q$ in $P$, where $|pq|$ denotes
	the Euclidean distance between $p$ and $q$. 	
	In this paper, we present a data structure of size $O(n^{5/3})$ such that
	for any two points in $P$, we can check in $O(n^{2/3})$ time 
	if there is a path in $G$ between the two points. 
	This is the first data structure for answering reachability queries 
	whose performance depends only on $n$ but not on the number of edges. 
\end{abstract}

\section{Introduction}
Consider a set $S$ of unit disks in the plane. The \emph{intersection graph}  
for $S$ is defined as the undirected graph whose
vertices correspond to the disks in $S$ such that 
two vertices are connected by an edge if and only if the two disks 
corresponding to them intersect. 
It can be used as a model for broadcast networks: The disks of $S$ represent transmitter-receiver  stations with the same
transmission power.  
One can view the broadcast range of a transmitter as a unit disk.

One straightforward way to deal with the intersection graph for $S$ is to construct 
the intersection graph explicitly, and then run algorithms designed for general graphs. 
However, the intersection graph for $S$ 
has complexity $\Theta(n^2)$ in the worst case even though it can be (implicitly) represented 
as $n$ disks. 
Therefore, it is natural to 
seek faster algorithms for an intersection graph implicitly represented as its underlying set of disks. 
For instance, the shortest path between two vertices in a unit-disk intersection graph can be computed in near linear time~\cite{wang2020near}. 
For more examples, refer to~\cite{cabello2015shortest,chan2019approximate,kaplan2018routing}.

\medskip
A \emph{transmission graph} is a \emph{directed} intersection graph, which is introduced to model broadcast networks in the case that
transmitter-receiver stations have different transmission power~\cite{peleg2010localized,von2009algorithmic}. 
Let $P$ be a set of $n$ points in the plane where each point $p$ of $P$ is associated with a 
radius $r_p>0$. 
The \emph{transmission graph} $G=(V,E)$ of $P$ is an weighted directed graph 
whose vertex set corresponds to $P$. There is an edge $\direct{p}{q}$ in $E$ for two points $p$ and $q$ in $P$ 
if and only if the Euclidean distance $|pq|$ between $p$ and $q$ is at most $r_p$.  The weight of an edge $\direct{p}{q}$ is defined as $|pq|$. 
It is sometimes convenient to consider a point $p$ of $P$ as the disk of radius $r_p$ centered at $p$.
We call it the \emph{associated disk} of $p$, and denote it by $D_p$. 
We say $p$ is \emph{reachable} to $q$ if there is a $p$-$q$ path in $G$. 

\medskip
In this paper, we consider the \emph{reachability} problem for transmission graphs:
Given a set $P$ of points associated with radii, check if a point of $P$ is reachable to another point of $P$ in the transmission
graph. 
In the context of broadcast networks, this problem asks if a transmission station can transmit information to a receiver. 
We consider three versions of the reachability problem: the single-source reachability problem,
(discrete) reachability oracles, and continuous reachability oracles. 
The \emph{single-source reachability problem} asks to compute all vertices reachable from a given source node $p\in P$ in
the transmission graph of $P$. Indeed, we consider the more general problem that asks to compute a $t$-spanner of size $O(n)$. 
Once we have a $t$-spanner of size $O(n)$, we can compute all vertices reachable from a given source node in linear time. 
A \emph{(discrete) reachability oracle} is a data structure for $P$ 
so that, 
given any two query points $p$ and $q$ in $P$, we can check if $p$ is reachable to $q$ in $G$
efficiently. 
A \emph{continuous reachability oracle} is a data structure for $P$ 
for answering reachability queries that takes two points in the plane, one in $P$ and one not necessarily in $P$,
as a query. 

\subsection{Previous Work.} 
The reachability problems and shortest-path problems have been extensively studied not only for general graphs but also
for special classes of graphs; directed planar graphs~\cite{holm2015planar}, Euclidean spanners~\cite{gudmundsson2008approximate,oh:LIPIcs:2020:13396}, and disk-intersection graphs~\cite{cabello2015shortest,chan2019approximate}.
In the following, we introduce several results for transmission graphs of disks in the plane. 
Let $\Psi$ be the ratio between the largest and the smallest radii associated with the points in  $P$. 
\begin{itemize}
	\item \textbf{$t$-Spanners (Single-source reachability problem).} 
	One can solve the single-source reachability problem for transmission graphs
	in $O(n\log^4 n)$ time by constructing a dynamic data structures
	for weighted nearest neighbor queries~\cite{chan:LIPIcs:2019:10428,kaplan_et_al:ACM-SIAM2017}.  
	Kaplan et al.~\cite{kaplan_et_al:SIAMJC2018} presented two algorithms for the more general
	problem that asks to compute a $t$-spanner of size $O(n)$ for any constant $t>1$,
	one with $O(n\log^4 n)$ time and one with $O(n\log n+ n\log \Psi)$ time. 
	Recently, Ashur and Carmi~\cite{ashurt} also considered this problem, and 
	presented an $O(n^2\log n)$-time algorithm for computing a $t$-spanner of which 
	every node  has a constant in-degree, and the total weight is bounded by a function of $n$ and $\Psi$. 
	Also, spanners for transmission graphs in an arbitrary metric space also have been considered~\cite{peleg2010localized,peleg2013relaxed}.
	\medskip 
	\item \textbf{Discrete reachability oracles.}
	Kaplan et al.~\cite{kaplan_et_al:LIPIcs:2015:5106} presented three reachability oracles: 
	one for $\Psi < \sqrt{3}$, two for an arbitrary $\Psi>1$. 
	For an arbitrary $\Psi$, their first reachability oracle
	has performance which polynomially depends on $\Psi$, and the second one
	has performance which polylogarmically depends on $\Psi$. 
	More specifically, the first data structure uses space $O(\Psi^3 n^{1/2})$, and has query time $O(\Psi^5 n^{3/2})$.
	The second one uses space $\tilde{O}_{n,\Psi}(n^{5/3})$, and has query time $\tilde{O}_{n,\Psi}(n^{2/3})$, where $\tilde{O}_{n,\Psi}$ hides
	polylogarithmic factors in $\Psi$ and $n$. This data structure is randomized in the sense that it allows to answer all queries correctly with high probability. 
	\medskip
	
	\item 
	\textbf{Continuous reachability oracles.}
	Kaplan et al.~\cite{kaplan_et_al:SIAMJC2018} shows that a discrete reachability oracle for the transmission graph $G$ of $P$ 
	can be extended to a continuous reachability oracle.
	More specifically, given a discrete reachability oracle for $G$ with space $S(n)$  and  query time $Q(n)$, 
	one can obtain in $O(n\log n\log \Psi)$ time a continuous reachability oracle for $G$ with  space 
	$S(n)+O(n\log\Psi)$ and query time $O(Q(n)+\log n\log\Psi)$.
\end{itemize}

\subsection{Our Results.} As mentioned above, we improve the previously best-known results of the three versions of the reachability problem for transmission graphs.
\begin{itemize}
	\item \textbf{$t$-Spanners (Single-source reachability problem).} 
	We first present an $O(n\log^3 n)$-time algorithm for computing a $t$-spanner for a constant $t>0$ in Section~\ref{sec:bfs}, 
	which improves the running time of the algorithm by \cite{kaplan_et_al:SIAMJC2018} by a factor of $O(\log n)$. 
	Our construction is based on the $\Theta$-graph and grid-like range tree introduced by~\cite{Agarwal:CCCG2013}. 
	This algorithm is also used for computing reachability oracles in Sections~\ref{sec:oracle} and~\ref{sec:geom}, and Section~\ref{sec:psioracle}. 
	\medskip 
	\item \textbf{Discrete reachability oracles.}
	We present two discrete reachability oracles for the transmission graph of $P$.
	The first one described in Section~\ref{sec:oracle} uses space $O(n^{5/3})$ and has query time $O(n^{2/3})$, and
	can be computed in $O(n^{5/3})$ time. 
	This is the first reachability oracle for a transmission graph whose performance is independent of
	$\Psi$.
	
	The second one is described in Section~\ref{sec:psioracle}. Its performance parameters depend  
	polylogarithmically on the radius ratio $\Psi$. More specifically, it
	uses space $\tilde{O}_\Psi(n^{5/2})$, and has query time $\tilde{O}_\Psi(n^{3/2})$.
	It can be constructed in  $\tilde{O}_\Psi(n^{5/2})$, where $\tilde{O}_\Psi(\cdot)$ hides polylogarithmic factors in $\Psi$. 
	To obtain this, we combine two reachability oracles given by \cite{kaplan_et_al:LIPIcs:2015:5106} whose performance parameters using a balanced separator  of smaller size introduced by~\cite{FoxPach}. 
	\medskip
	\item 
	\textbf{Continuous reachability oracles.}
	We also present a \emph{continuous reachability oracle} with space $O(n^{5/3})$, query time $O(n^{2/3})$, and
	preprocessing time $O(n^{5/3}\log^2 n)$ in Section~\ref{sec:geom}, 
	which is the first continuous reachability oracle whose performance is independent of $\Psi$.
	Instead of using the approach in~\cite{kaplan_et_al:SIAMJC2018}, we 
	use auxiliary data structures whose performance is independent of $\Psi$ 
	together with the reachability oracle described in Section~\ref{sec:oracle}. 
\end{itemize}

\section{Improved Algorithm for Computing a \texorpdfstring{$t$}{t}-Spanner}\label{sec:bfs}
Let $P$ be a set of $n$ points associated with radii, and $G=(P,E)$ be the transmission graph
of $P$. 
A subgraph $H$ of $G$ is called a \emph{t-spanner} of $G$ if
 for every pair of vertices of $G$, 
the distance in $H$ between them is at most $t$ times the distance in $G$ between them. 
A \emph{sparse} $t$-spanner is 
useful for constructing a reachability oracle efficiently;
a $t$-spanner preserves the reachability information of $G$, and
it allows us to investigate a small number of edges.  
Therefore, we first consider the problem of constructing a $t$-spanner of $G$ in this section,
and we use it for constructing a reachability oracle in Section~\ref{sec:oracle}.  

In this section, we present an $O(n\log^3 n)$-time algorithm for computing a $t$-spanner of $G$ of size $O(n)$ 
for any constant $t>1$. 
This improves the running time of the algorithm proposed by Kaplan et  al.~\cite{kaplan_et_al:LIPIcs:2015:5106}, which runs in $O(n\log^4 n)$ time.\footnote{ 
	Kaplan et al. mentioned that this algorithm takes an $O(n\log^5 n)$ time. However, this can be improved automatically into $O(n\log^4 n)$ using a data structure of~\cite{chan:LIPIcs:2019:10428}.}
The spanner constructed by Kaplan et al. is a variant of 
the Yao graph. They first show that a variant of the Yao graph is a $t$-spanner
for $G$, and then show how to construct it efficiently.


\subsection{Theta Graphs and \texorpdfstring{$t$}{t}-Spanners of Transmission Graphs}\label{sec:theta}
Our spanner construction is based on the \emph{$\Theta$-graph}, which
 is a geometric spanner similar to the Yao graph. 
Let $k>0$ be a constant, which will be specified later, depending on $t$. 
Imagine that we subdivide the plane into $k$ interior-disjoint cones with opening angle $2\pi/k$ which have the origin
as their apexes. 
Let $\mathcal{F}$ be the set of such cones. See Figure~\ref{theta-4}(a).
For a cone $F\in\mathcal{F}$ and a point $p\in P$, let $F_p$ denote the translated cone of $F$ so that the apex of $F_p$ lies on $p$. 
For each point $p\in P$, we pick $k$ incoming edges for $p$, one for each cone of $\mathcal{F}$, as follows. 
\medskip

For a point $q$ contained in $F_p$, let $q_\ell$  denote the 
the orthogonal projection of $q$ on
the angle-bisector of $F_p$. 
Also, we let $d_F(p, q)$ be the Euclidean distance between $p$ and $q_\ell$,
and let $\nn{F}{p}$ denote the point $q$ in $F_p$ with $\direct{q}{p}\in E$ that minimizes $d_F(p,q)$. 
 See Figure~\ref{theta-4}(b).
Note that $\nn{F}{p}$ might not exist.
%
For each cone $F\in \mathcal{F}$ and each point $p\in P$, 
we choose $\direct{\nn{F}{p}}{p}$. 
See Figure~\ref{theta-4}(c).
Let $H_k$ be the graph consisting of the points in $P$ and the chosen edges. 
If it is clear from the context, we simply use $H$ to denote $H_k$.

\begin{figure}[t]
\centering
   \includegraphics[width=0.8\textwidth]{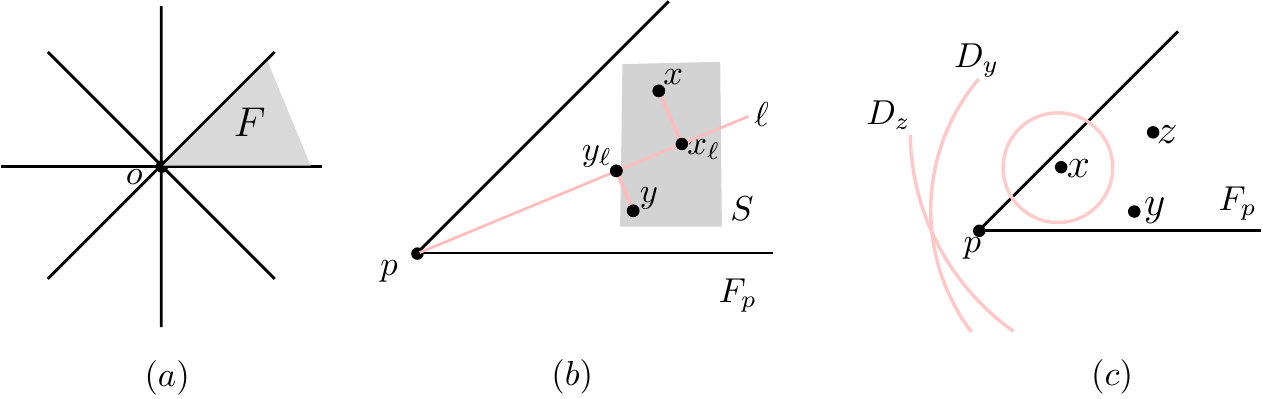}
   \hfil
\caption{Theta graph construction for $k=8$. (a) The $k$ cones of $\mathcal{F}$ subdivides
	the plane. (b) $\nn{F}{p}= y$, and $\nn{S}{p}=y$. (c) The edge $\direct{y}{p}$ is picked.}
\label{theta-4}
\end{figure}

To show that the $H_k$ forms a $t$-spanner, we need the following technical lemma. 

\begin{lemma} \label{theta:0}
	For a point $p$ in $P$ and a cone $F$ in $\mathcal{F}$, consider two points 
	$u$ and $v$ contained in $F_p$ such that $\direct{u}{p}\in E$ and $\direct{v}{p}\in E$. 
	Suppose the opening angle of $F$ is smaller than $\pi/3$, that is, $k>6$. 
	If $d_F(p,v)<d_F(p,u)$, then $\direct{u}{v}\in E$ and $|uv|<|up|$.
\end{lemma}

\begin{proof}
	Consider the triangle bounded by the boundary of $F_p$ and 
	a line through $u$ orthogonal to the angle-bisector of $F_p$.
	Notice that this triangle is an isosceles triangle containing $v$. 
	Since the top angle of the triangle is smaller than $\pi/3$, the apex $p$ is the farthest point from $u$ within the triangle.
	This implies $|uv|<|up|$. Since $r_u$ is at least $|up|$, the edge $\direct{u}{v}$ is contained in $E$.
\end{proof}

\begin{lemma}
\label{theta:1}
For an integer $k>8$, $H_k$ is a $\tan(\frac{\pi}{4}+\frac{2\pi}{k})$-spanner of $G$.
\end{lemma}

\begin{proof}
	We want to show that for every edge $e=\direct{u}{p}$ in $G$, there is a path in $H_k$ from $u$ to $p$ whose length is at most $\tan(\frac{\pi}{4}+\frac{2\pi}{k})\cdot|up|$. Let $t=\tan(\frac{\pi}{4}+\frac{2\pi}{k})$.
	
	To show this, we use the induction on the length of the edges. 
	For the base case, assume that $\direct{u}{p}$ is the shortest edge of the transmission graph. 
	Let $F$ be the cone of $\mathcal{F}$ such that $F_{p}$ contains $u$. 
	By construction, the directed edge from $\nn{F}{p}$ to $p$ is an edge of $H$. 
	Let $v=\nn{F}{p}$. 
	If $u=v$, then $\direct{u}{p}$ is an edge of $H$, and thus we are done. 
	Otherwise, $\direct{u}{v}$ is an edge of the transmission graph $G$ by Lemma~\ref{theta:0},
	and moreover, it is shorter than $\direct{u}{p}$, which contradicts that $\direct{u}{p}$ is the shortest
	edge of $G$.

	
	Now consider an edge $\direct{u}{p}$, and suppose that for every edge in the transmission graph shorter than $\direct{u}{p}$, there is a path connecting them whose length is at most $t$ times 
	their Euclidean distance. 
	Let $F$ be the cone of $\mathcal{F}$ such that $F_{p}$ contains $u$. 
	By construction, the directed edge from $\nn{F}{p}$ to $p$ is an edge of $H$. 
	Let $v=\nn{F}{p}$. 
	If $u=v$, then $\direct{u}{p}$ is an edge of $H$, and thus we are done. 
	Thus in the following, we assume that $u\neq v$.
	In this case, $\direct{u}{v}$ is an edge of the transmission graph $G$ by Lemma~\ref{theta:0},
	and moreover, it is shorter than $\direct{u}{p}$.
	
	Therefore, there is a path $\pi$ from $v$ to $p$ whose length is at most $t|vp|$ by the induction  hypothesis. Since $\direct{u}{v}$ is an edge of $G$, 
	the concatenation of $\pi$ and $\direct{u}{v}$ is a path of $G$ whose length is at most  $|uv|+t|vp|$. 
	Let $s$ be the projection point of $v$ to $\direct{u}{p}$. 
	See Figure~\ref{theta-1}. 
	We consider two cases with respect to the position of $s$.
	
	Since $t=\tan(\frac{\pi}{4}+\frac{2\pi}{k}) = \frac{1+\tan(\frac{2\pi}{k})}{1-\tan(\frac{2\pi}{k})}$, we have    
	$\frac{t-1}{t+1} = \tan(\frac{2\pi}{k}) > \tan(\frac{\pi}{k})$.

	\begin{figure}
		\centering
		\includegraphics[width=0.65\textwidth]{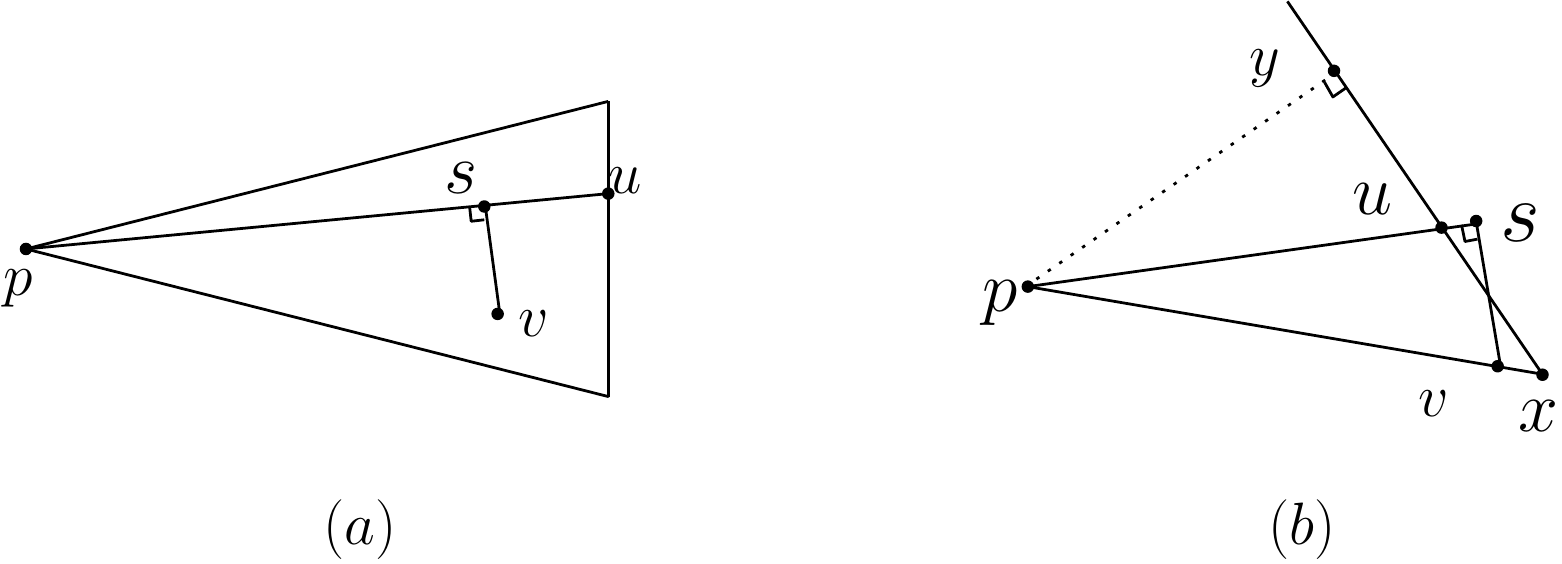}
		\hfil
		\caption{Two cases of the location of $s$}
		\label{theta-1}
	\end{figure}
	
	\paragraph{Case 1.}
	Suppose $s$ lies on $\direct{u}{p}$. See Figure~\ref{theta-1}(a). Since $u, v$ and $s$ are contained in $F_p$, 
	we have $\angle upv\leq 2\pi/k$. Then, 
	\begin{align*}
	|vp|+t|uv| &< |ps|+|sv|+ t(|us|+|sv|) \\
	&=  t(|us|+|ps|) +(t+1) |sv|+(1-t)|ps|.
	\end{align*}
	Note that $|sv|=|ps|\tan(\angle upv)$ and $\tan(\angle upv)\leq\tan(2\pi /k)=\frac{t-1}{t+1}$. 
	We obtain $(t+1)|sv|\leq(t-1)|ps|$. Then, 
	\begin{align*}
	t(|us|+|ps|) +(t+1) |sv|+(1-t)|ps| &\leq t(|us|+|ps|)\\
	&=t|up|
	\end{align*}
	\paragraph{Case 2.}
	Now suppose $s$ does not lie on $up$. See Figure~\ref{theta-1}(b) for illustration. 
	Let $x$ be an intersection point of the line through $pv$ and the line $l$ that passes $u$ which is orthogonal to the angle bisecting line of the cone. 
	Also, let $y$ be a projection point of $p$ into $l$.
	Similarly, by Lemma~\ref{theta:0}, there is a path $\pi$ from $u$ to $v$ such that the length of $\pi$ is less than $t|uv|$ by the induction hypothesis.
	Also, we consider the concatenation of $\pi$ and edge $vp$.
	Then, the length of this path is 
	\begin{align*}
	|vp| + t|uv| &\leq |px| + t|ux| \\
	&< |py| + (t+1)|yx| \\
	&< (1+(t+1)\tan(\pi /k))|py| \ \ (\because \tan(\angle yps) \leq \tan(\pi /k))   \\
	&< (1+(t+1)\tan(\pi /k))|up| \\
	&\leq t|up| \ \ (\because \tan(\pi /k) < (t-1)/(t+1)).
	\end{align*}
	
	\paragraph{}
	Therefore, for any case, there is a path from $u$ to $p$ with its length at most $t|up|$. This completes the proof.
\end{proof}

Note that $t=\tan(\frac{\pi}{4}+\frac{2\pi}{k})>1$ converges to $\tan(\frac{\pi}{4})=1$ as $k\rightarrow \infty$. 
Therefore, 
for any constant $t>1$, we can find a constant $k$ such that $H_k$ is a $t$-spanner of the transmission graph. 

\subsection{Efficient Algorithm for Computing the \texorpdfstring{$t$}{t}-Spanner} \label{sec:range}
In this section, we give an $O(n\log^3 n)$-time algorithm to construct $H_k$ for a constant  $k>6$. 
To compute all edges of $H_k$, for each point $p\in P$ and each cone $F\in\mathcal F$, consider the translated cone $F_p$ of $F$ so that the apex lies on $p$, and compute $\nn{F}{p}$. 
We show how to do this for a cone $F\in \mathcal F$ only. 
The others can be handled analogously.
Without loss of generality, we assume that the counterclockwise angle from the positive $x$-axis
to two rays of $F$ are $0$ and $2\pi/k$, respectively. 
Let $\ell_1$ and $\ell_2$ be two lines orthogonal to the two rays, respectively.

\paragraph{Approach of Kaplan et al.} 
The spanner constructed by Kaplan et al.~\cite{kaplan_et_al:LIPIcs:2015:5106} is a variation of the Yao graph. For each cone  $F\in\mathcal{F}$ and a point $p\in P$,
they pick the closest point in $F_p$ to $p$ among all points $q$ with $p\in D_q$. 
Since they choose the closest point in a cone with respect to the Euclidean distance, 
they need to fit grid cells into a cone. To resolve this, they use various data structures
 including
a compressed quadtree, a power diagram, a well-separated pair decomposition, and a dynamic nearest neighbor search data structure.

\paragraph{Our Approach.}
Instead, our construction is based on the $\Theta$-graph. 
Recall that we pick the closest point in a cone with respect to $d_F(\cdot,\cdot)$ instead of 
the Euclidean distance. 
The order of the points of $F_p\cap P$ sorted with respect to $d_F(p, \cdot)$ is indeed
the order of them sorted with respect to their projection points onto the angle-bisector of $F$. 

In the following, we present an $O(n\log^3 n)$-time algorithms for computing all edges of $H_k$ constructed for $F$.  
To do this, we use grid-like range trees proposed by Moidu. et al.~\cite{Agarwal:CCCG2013} together with a power diagram. 
With a slight abuse of notation, for a region $S$ contained in $F_p$,  
let $\nn{S}{p}$ be the point $q$ of $S\cap P$ with $\direct{q}{p}\in E$ that minimizes $d_F(p,q)$. 
See Figure~\ref{theta-4}(b). 

\subsubsection{Data structures.}


We construct  
the \emph{two-level grid-like range tree} introduced by Moidu et al. \cite{Agarwal:CCCG2013} with respect to $\ell_1$ and $\ell_2$.
It is a two-level balanced  
binary search tree. 
The first-level tree $T_1$ is a balanced  
binary search tree on the $\ell_1$-projections of the points of $P$.  
Each node $\alpha$ in the first-level tree corresponds to a slab $I(\alpha)$ orthogonal to $\ell_1$.   
It is also associated with the
second-level tree $T_{\alpha}$ which is a binary search tree, not necessarily balanced, 
on the points of $P\cap I(\alpha)$. Unlike the standard range 
tree~\cite{CGbook}, $T_\alpha$ is obtained from a balanced binary search tree $T_2$ on the  $\ell_2$-projections of the points of $P$. 
More specifically, we remove the subtrees rooted at
all nodes of $T_2$ whose corresponding parallelograms contain no point in $P\cap I(\alpha)$ in their union, 
and contract all nodes which have only 
one child. 
Then $T_\alpha$ is not necessarily balanced but a full binary tree of depth $O(\log n)$.

Given a point $p$ of $P$, 
there are $O(\log^2 n)$ interior-disjoint parallelograms whose union contains all points of $P\cap F_p$. 
We denote the set of these parallelograms by $\mathcal B_p$. 
By construction, the
cells of $\mathcal{B}_p$ are aligned for any point $p\in P$
so that we can consider them as a grid of size $O(\log n)\times O(\log n)$.  
See Figure~\ref{theta-3}.

\begin{lemma}[\cite{Agarwal:CCCG2013}]\label{lem:grid-like}
	The two-level grid-like range tree on a set of $n$ points in the plane
	can be computed in $O(n\log n)$ time. Moreover, its size is $O(n\log n)$.
\end{lemma}

Then for each node $v$ of the second-level trees, we construct
a balanced binary search tree of the $\ell$-projections of $P\cap B(v)$ as the third-level tree,
where $\ell$ denotes the angle bisector of $F$.  
For a node $\beta$ of the third-level trees, let $P(\beta)$ denote the set of the points stored in the subtree rooted at $\beta$.  
we construct the power diagram of $P(\beta)$. The \emph{power diagram} is a weighted version of the Voronoi diagram. 
More specifically, the \emph{power distance} between a point $p$ and a disk $D_q$ is defined as $|pq|^2-r_q^2$. 
The power diagram partitions the plane into $n$ regions such that all points in a same region have the same closest disk in power distance.
The power diagram of $n$ disks can be constructed in $O(n\log n)$ time with $O(n)$ space. 
Also, we can locate the disk $D$ that minimizes the power distance from a query point $p$ in $O(\log n)$ time. 
As a consequence, we can determine in $O(\log n)$ time if the query point $p$ is in the union of disks by checking if $p\in D$~\cite{imai1985voronoi,kaplan_et_al:ACM-SIAM2017}. 

%

The construction time of the first, second, and third-level trees is $O(n\log^3 n)$ in total.  
Then we construct the power diagram for each node of a third-level tree in  a \emph{bottom-up} fashion.
In particular, we start from constructing the power diagrams of the leaf nodes. 
For each internal node, we compute its power diagram by merging the power diagram of its two children. 
Therefore, we can construct the power diagrams for all nodes of a third-level tree in 
$O(m\log m)$ time, where $m$ denotes the number of points corresponding to the root of the third-level tree. 
Since the sum of $m$'s over all third-level trees
is $O(n\log^2n)$, the whole data structure can be constructed in $O(n\log^3 n)$ time.

\subsubsection{Query algorithm.}
%
For each cell $B\in\mathcal{B}_p$, 
we compute $\nn{B}{p}$ in $O(\log^2 n)$ as follows.
We start from the root of the third-level tree associated with $B$. 
We check if there is a point $q\in P(\beta)$ with $\direct{q}{p}\in E$ using the power diagram stored in the root node. 
If it does not exist, $\nn{B}{p}$ does not exist.
Otherwise, we traverse the third-level tree until we reach a leaf node.  
For each node $\beta$ we encounter during the traversal, we consider the left child of $\beta$, say $\beta_L$. 
We check if there is a point $q\in P(\beta_L)$ with $\direct{q}{p}\in E$ using the power diagram stored in $\beta_L$. 
If it exists, we move to $\beta_L$. Otherwise, we move to the right child of $\beta$. We do this until we reach a leaf node, which 
stores  $\nn{B}{p}$. 
%
%

In the following, we show how to choose $O(\log n)$ cells of $\mathcal{B}_p$, one of which contains 
$\nn{B}{p}$.
The cells of $\mathcal{B}_p$ are aligned along $\ell_1$ and $\ell_2$.
They can be considered as a grid of $O(\log n)\times O(\log n)$ cells. We represent each row (parallel to $\ell_1$) by integers $1,\ldots,O(\log n)$, and each column (parallel to $\ell_2)$ by integers $1,\ldots,O(\log n)$.  
We represent each cell of $\mathcal{B}_p$ by a pair $B(i,j)$ of indices such that $i$ is the row-index and $j$ is the column-index of the cell. 
For illustration, see Figure~\ref{theta-3}(a). 
A cell $B=B(i,j)$ is said to be \emph{useful} if $\nn{B}{p}$ exists. 
Also, a useful cell $B=B(i,j)$ is called an \emph{extreme cell} of $\mathcal{B}_p$ if 
no cell $B(i',j')$ is useful for indices $i'$ and $j'$ such that $i-j = i'-j'$ and $i' <i$. 

\begin{lemma}\label{lem:extreme}
	The cell of $\mathcal{B}_p$ containing $\nn{F}{p}$ is an extreme cell. Moreover, the number of extreme cells of $\mathcal{B}_p$ is $O(\log n)$. 
\end{lemma}

\begin{proof}
	Suppose $\nn{B}{p}$ exists for $B=B(i,j)$. For each integer $w>0$, $B(i+w, j+w)$ 
	lies in the upper right part of $B(i,j)$.
	This means that $d_F(p,q)$ is at least  $d_F(p,q')$ for any two points $q\in B(i+w,j+w)$ and $q'\in B(i,j)$.
	Thus, the cell containing $\nn{F}{p}$ is an extreme cell.
	
	The number of extreme cells is equal to the number of distinct $i-j$ values among the cells $B(i,j)\in \mathcal{B}_p$.
	This number is $O(\log n)$ since each index of rows and columns is a positive integer at most $O(\log n)$. 
\end{proof}

\begin{figure}
	\centering
	\includegraphics[width=0.7\textwidth]{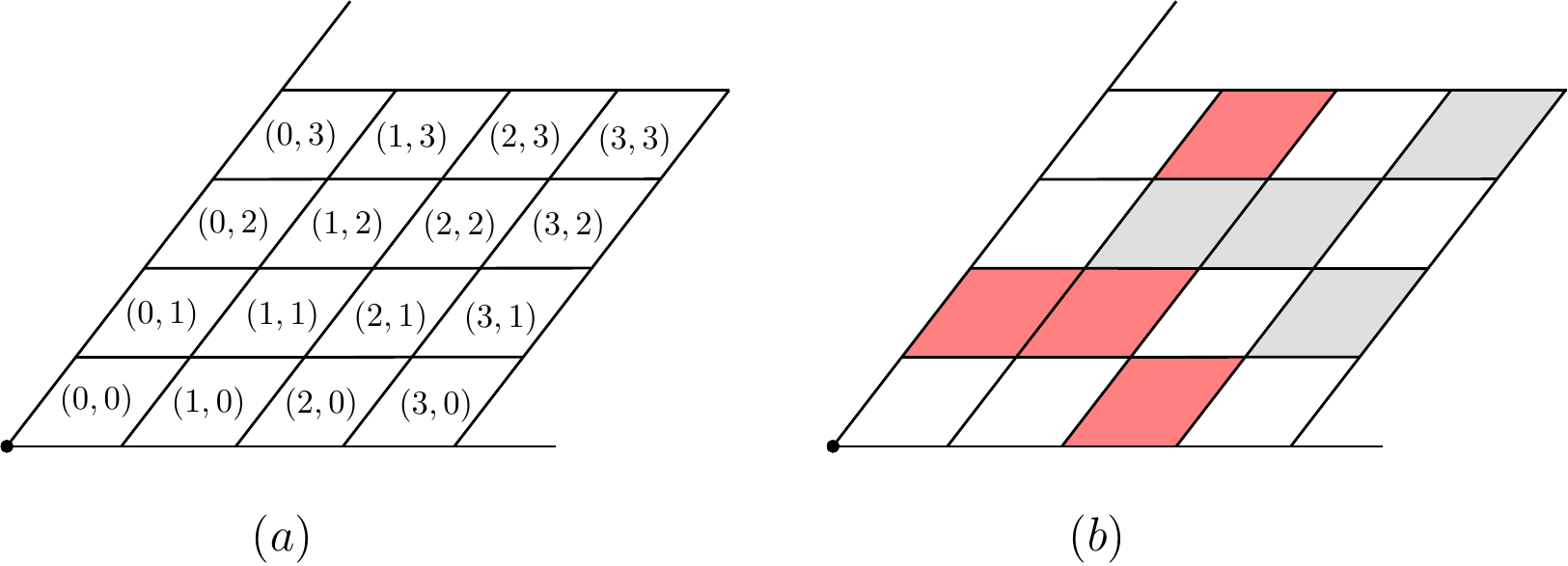}
	\caption{(a) Index of the grid-like range tree  (b) Useful cells are colored with gray or red, and extreme cells are colored red.}
	\label{theta-3}
\end{figure}

To compute $\nn{F}{p}$, we first compute $\mathcal{B}_p$ in $O(\log^2 n)$ time. 
For each cell $B\in\mathcal{B}_p$, we check if it is useful using the power diagram of $P\cap B$, which is stored
in the root node of the third-level tree in $O(\log^3 n)$ time in total.  
Then we choose $O(\log n)$ extreme cells among the useful cells of $\mathcal{B}_p$.
For each cell $B$ of them, we compute $\nn{B}{p}$ in $O(\log^2 n)$ time, and thus the total query time is $O(\log^3 n)$.

\begin{theorem} \label{theta:2}
Given a point set $P$ and a constant $t>1$, we can construct a $t$-spanner of the transmission graph of $P$ within $O(n\log^3 n)$ time.
\end{theorem}

\subsection{Computing a BFS Tree Using a \texorpdfstring{$t$}{t}-Spanner}\label{apnd:bfs}
In this section, we construct a BFS tree for the transmission graph $G$.
For a root $s$, a \emph{BFS tree} is a shortest-path tree of $G$ rooted at $s$ where the length of a path 
is measured by the number of edges in the path.  

Kaplan et al. proposed an algorithm of constructing a BFS tree via their $t$-spanner that is a variant of the Yao Graph.
They utilize the technique proposed by Cabello et al. \cite{cabello2015shortest} for their algorithm and correctness.
Our algorithm is exactly the same as the algorithm by Kaplan et al. However, since our $t$-spanner is based on the $\Theta$-graph instead of the Yao Graph,   we have to prove the correctness. 

For a given root $s$ of a BFS tree, 
let $W_i$ be the set of points in $P$ with depth $i$. 
The correctness of the algorithm follows from the following lemma.

\begin{lemma} \label{bfs:1}
	Let $H$ be a $t$-spanner as in Theorem \ref{theta:2}, and let $v\in W_{i+1}$. Then, there is a path $\pi=up_k\ldots p_1v$ in $H$ with
	$u\in W_i$ and $p_j\in W_{i+1}$ for all indices $j\in[1,k]$.
\end{lemma}
\begin{proof}
	For a point $v\in W_{i+1}$,  let $s(v)$ denote the smallest Euclidean distance between $v$ and a point $u\in W_i$ such that $\direct{u}{v}\in E$. 
	We prove the lemma using induction on $s(\cdot)$ for the points in $W_{i+1}$. 
	
	For the base case, let $v\in W_{i+1}$ have the smallest value of $s(v)$.
	Let $u\in W_i$ be the point with $s(v)=|uv|$. 
	If $u=\nn{F}{v}$ for a cone $F$ of $\mathcal{F}$, we are done. Otherwise, we consider 
	the cone $F\in\mathcal{F}$ such that $F_v$ contains $u$, and we consider $w=\nn{F}{v}$.
	By Lemma \ref{theta:0}, $\direct{u}{w}$ and $\direct{w}{v}$ are edges of the transmission graph.
	Also, $|uw|$ is less than $|uv|=s(v)$. 
	Let $k$ be the index with $w\in W_k$. We have $k\leq i+1$ since $u\in W_i$ and $\direct{u}{w}\in E$. 
	Moreover, $k \neq i+1$. (Otherwise, $s(w)<s(v)$, which makes contradiction.)
	Also, $k\geq i$ since $\direct{w}{v}\in E$ and $v\in W_{i+1}$. 
	Therefore, $\pi=wv$ is desired path because $w\in W_i$.
	
	Now, we consider a node $v\in W_{i+1}$, and suppose that all $v'\in W_{i+1}$ with $s(v')<s(v)$ satisfy the condition.  
	Similarly, if $u=\nn{F}{v}$, we are done. 
	Otherwise, we consider $w=\nn{F}{v}$  where $F_v$ contains $u$. The same properties hold by Lemma~\ref{theta:0}. 
	In particular, $w\in W_k$ with $k=i$ or $k=(i+1)$. 
	If $k=i$, we are done. 
	Otherwise, $s(v)=|uv|>|uw|\geq s(w)$. 
	Now, there is a path $\pi=up_k...p_1w$ with $u\in W_i$ and $p_j\in W_{i+1}$ for all $j$ due to the induction hypothesis on $w$. 
	Then, the path $\pi'=up_k...p_1wv$ satisfies the condition. This completes the induction.
\end{proof}



Cabello et al. proposed a BFS tree algorithm for unit-disk graphs by considering the edges of the Delaunay triangulation of the point set.  
Later, Kaplan et al. proposed a $t$-spanner based on a variation of the Yao graph. 
Their t-spanner provides similar properties for transmission graphs as the Delaunay triangulation does for unit-disk graphs. 
Our $t$-spanner also satisfies this property by Lemma~\ref{bfs:1}.
Indeed, Lemma~\ref{bfs:1} is the same as \cite[Lemma 1]{cabello2015shortest} except that $H$ is our spanner in Lemma~\ref{bfs:1} and it
is the Delaunay triangulation in \cite[Lemma 1]{cabello2015shortest}.
%
Then, we are able to reuse the algorithm of Cabello et al.
We remark that this algorithm takes $O(n\log n)$ time.

\begin{theorem} \label{bfs:2}
Let $P$ be a set of $n$ points, each associated with a radius.
Given a $t$-spanner $H$ of the transmission graph $G$ of $P$ as in Theorem~\ref{theta:2},  
we can construct a BFS tree of $G$ within $O(n\log n)$ time.
\end{theorem}

\section{Reachability Oracle for Unbounded Radius Ratio}\label{sec:oracle}
In this section, we present a data structure of size $O(n^{5/3})$ so that given any two points $p$ and $q$ in $P$, we can check if $p$ is reachable from $q$ in $O(n^{2/3})$ time. 
Moreover, this data structure can be constructed in $O(n^{5/3})$ time. Note that this result is independent to the radius ratio $\Psi$.


We say a set of disks is $k$-\emph{thick} if for any point $p$ in the plane, there are at most $k$ disks that contains $p$. 
Similarly, we say a transmission graph is $k$-\emph{thick} if its underlying disk set is $k$-thick. 

\begin{lemma}[{\cite[Theorem 5.1]{Teng:MiTeVa91}}] \label{4:separator}
 For any set $\mathcal D$ of disks that is $k$-thick, there is a circle $S$ intersecting $O(\sqrt{kn})$ disks of $\mathcal D$ such that  the number of disks of $\mathcal D$ with $|S_\textnormal{in}|, |S_\textnormal{out}|\leq \frac{2n}{3}$, 
 where $S_\textnormal{in}$ and $S_\textnormal{out}$ denote the set of disks of $\mathcal D$ contained in the interior of $S$ and
 the exterior of $S$, respectively. 
 In this case, We call $S$ a \emph{separating circle}. 
 Moreover, we can compute $S$, $S_\textnormal{in}$ and $S_\textnormal{out}$ in linear time. 
\end{lemma}

Consider a separating circle $S$ of the disk set induced by $P$. 
By Lemma~\ref{4:separator}, $P$ is partitioned into three sets $\sinn, \sout$, 
 and $S_\textnormal{cross} =\{p \in P \mid D_p \cap S \neq\emptyset \}$ such that  
 every path in $G$ connecting a point of $\sinn$ and a point of $\sout$ 
 visits a point in $S_\textnormal{cross}$.
We call $S_\textnormal{cross}$ a \emph{separator} of $G$ (or $P$).
Using separators, we build a \emph{separation tree} by repeatedly applying the algorithm in Lemma~\ref{4:separator}.
As we will see in Section~\ref{sec:remain}, the separation tree enables us to construct a reachability oracle efficiently. 
However, the transmission graph of a set of $n$ points is $n$-thick in the worst case, and in this case, 
Lemma~\ref{4:separator} does not give a non-trivial bound. 

To resolve this, we partition $P$ into $O(n^{2/3})$ \emph{chains}, each consisting of $O(n^{1/3})$ points of $P$, and 
the \emph{remaining set} $R$ of points of $P$ not belonging to any of the chains. Then we show that $R$ is $O(n^{1/3})$-thick, and thus
Lemma~\ref{4:separator} gives an efficient reachability oracle for the subgraph of $G$ induced by $R$. 
Additionally, 
we construct an auxiliary data structure for each chain. 

\begin{figure}
\centering
   \includegraphics[width=0.7\textwidth]{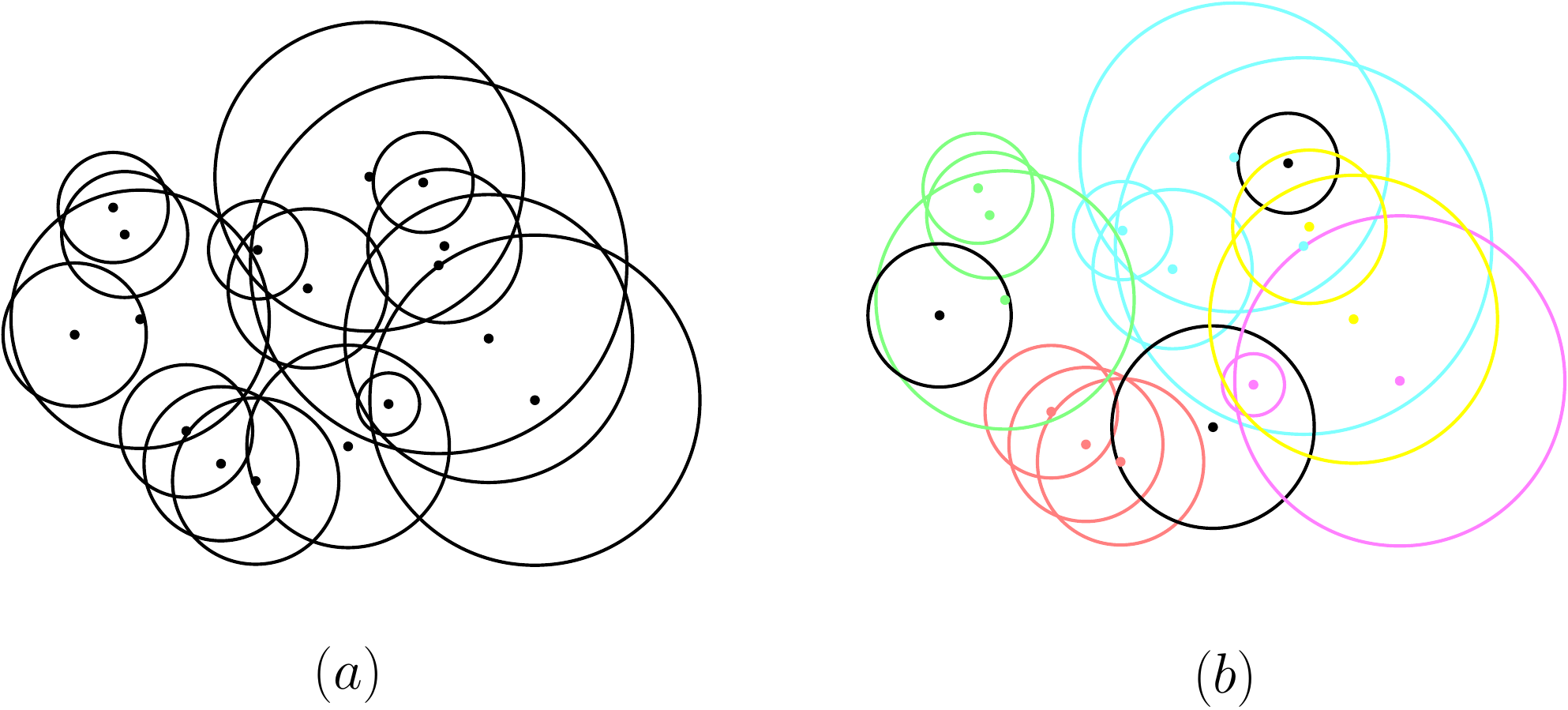}
\caption{\label{chain-1} (a) A set $P$ of points associated with radii.  (b) 
	The disks in the same chain are colored with the same color, and the points in $R$ are colored black.}
\end{figure}

\subsection{Chain}
 We call a sequence $\langle p_1, \ldots, p_k\rangle$ of points of $P$ sorted in the ascending order of their associated radii 
  a \emph{chain} if $\direct{p_j}{p_i}\in E$ for all indices $i$ and $j$ with $1\leq i < j$. 
 In other words, $|p_ip_j|\leq r_{p_j}$. 
 In this section, we construct $O(n^{1/3})$-length chains as many as possible 
 so that the remaining set $R$ is $k$-thick for a small $k$.

To compute chains, we need a dynamic data structure for a set $\mathcal D$ of disks, dynamically changing by insertions and deletions, 
such that for a query point, we can check if there is a disk of $\mathcal D$ that contains the query point. 
This can be obtained using dynamic 3-D halfspace lower envelope data structure, which
is given by \cite{chan:LIPIcs:2019:10428}, together with the standard lifting transformation. 
In particular, this data structure can be built in $O(n\log n)$ time and its insertion time, deletion time and query time are $O(\log^2 n)$, $O(\log^4 n)$ and $O(\log^2 n)$, respectively. 
For the convenience, we denote this data structure by $\mathsf{DNN}(\mathcal D)$.

\begin{lemma}\label{chain:1}
Let $\mathcal{D}$ be a set of disks, and $p$ be a point in the plane.
Given $\mathsf{DNN}(\mathcal D)$, we can check if there are $n^{1/3}$ disks of $\mathcal{D}$ containing $p$ 
in $O(n^{1/3}\log^4 n)$ time. 
Moreover, if they exist, we can return them, and delete them from $\mathcal D$ and $\mathsf{DNN}(\mathcal D)$ within the same time bound.
\end{lemma}

\begin{proof}
	We find a disk containing $p$ in $\mathcal{D}$ using $\mathsf{DNN}(\mathcal D)$, and 
	remove the returned disk from $\mathcal D$ and $\mathsf{DNN}(\mathcal D)$. 
	Note that $\mathsf{DNN}(\mathcal D)$ returns a disk if and only if there is a disk in $\mathcal{D}$ that contains the query point. 
	We repeat this $n^{1/3}$ times. 
	If $n^{1/3}$ distinct disks are returned, then we are done.
	Otherwise, there are less than $n^{1/3}$ disks that contain $p$. 
	In this case, we are required to insert all removed points to $\mathcal{D}$ and $\mathsf{DNN}(\mathcal D)$. 
	This procedure applies $O(n^{1/3})$ queries, insertions and deletions, so it takes $O(n^{1/3}\log^4 n)$ time in total.
\end{proof}




Let $\mathcal{D}$ be the set of disks induced by $P$, and we construct $\mathsf{DNN}(\mathcal D)$. 
We choose the smallest disk $D_p$ of $\mathcal D$ and remove $D_p$ from $\mathcal{D}$. Then we update $\mathsf{DNN}(\mathcal D)$
accordingly. We check if there are $n^{1/3}$ disks of $\mathcal{D}$ containing the center $p$ 
of $D_p$ by applying the algorithm in Lemma~\ref{chain:1}. If it returns $n^{1/3}$ disks,
let $L_p$ be the set consisting of $p$ and the centers of those disks. 
Since $\mathcal{D}$ is updated, we can apply this procedure again. We do this until $\mathcal{D}$ is empty. 
As a result of this repetition, we obtain sets $L_p$'s of points of $P$. Note that the disks induced by $L_p$
contain $p$, and the number of $L_p$'s is $O(n^{2/3})$. 

Next, for each set $L_p$, we consider six interior-disjoint cones with opening angle $\pi/3$ with apex $p$. 
For each cone $F$, we sort the points of $L_p\cap F$ in the ascending order of their associated radii.
Then we claim that the sorted list is a chain, and thus we obtain six chains for each set $L_p$. 
Therefore, we have $O(n^{2/3})$ chains in total. 

\begin{lemma}\label{lem:chain}
 The sequence of the points of $L_p\cap F$ sorted in the ascending order of their associated radii is a chain.
\end{lemma}

\begin{proof}
	Consider two points  $x$ and $y$ in the given sequence such that $x$ lies before $y$ in the sequence.  
	Note that $r_x \leq r_y$ by construction. 
	We show that $\direct{y}{x}\in E$. 
	Let $\triangle_y$ be the regular triangle surrounded by the two rays of $F$ and a line passing through $y$. 
	Then, $p$, one corner of $\triangle_y$, is one of the farthest points from $y$ within $\triangle_y$. 
	Thus, if $x$ is contained in $\triangle_y$, we have $r_y\geq |yp|$ and $r_y\geq |yx|$, and thus 
	$\direct{y}{x}\in E$. 
	Otherwise, $\triangle_y \subset \triangle_x$, and thus $\triangle_x$ contains $y$. 
	In this case, since $p$ is one of the farthest points from $x$ in $\triangle_x$, 
	$|yx|\leq |xp| \leq r_x \le r_y$, and thus $\direct{y}{x}\in E$. 
\end{proof}

Therefore, we have a set $\mathcal{C}$ of $O(n^{2/3})$ chains of length $O(n^{1/3})$.
We call the set of points of $P$ not contained in any of the chains of $\mathcal{C}$
the \emph{remaining set}. 
Also, we use  $\mathcal{R}$ to denote the subgraph of $G$ induced by $R$, and call it
the \emph{remaining graph}.

\begin{lemma} \label{4:chain,thick}
The graph $\mathcal R$ is $6n^{1/3}$-thick.
\end{lemma}
\begin{proof}
We first claim that the remaining set $R$ does not have a $n^{1/3}$-length chain.
Assume to the contrary that there is a $n^{1/3}$-length chain $C$, and let $p$ be the first point in $C$. 
At some moment in the course of the algorithm, $D_p$ becomes the smallest disk of $\mathcal{D}$.
At this moment, all disks associated with the points in $C$ are contained in $\mathcal{D}$. That is, 
at least $n^{1/3}$ disks of $\mathcal{D}$ contain $p$, and thus $p$ must be contained in a chain of $\mathcal{C}$, which contradicts
that $p$ is a point of $R$. 

Then we show that $R$ is $6n^{1/3}$-thick.
For any point $x$ in the plane, we consider 
six interior-disjoint cones with opening angle $\pi/3$ with apex $x$. 
For a cone $F$, consider the list $L$ of the points $p$ of $R\cap F$ with $r_p \geq |px|$ sorted in the ascending order of their associated radii. 
The proof of Lemma~\ref{lem:chain} implies that $L$ is a chain. 
By the claim mentioned above, the size of $L$ is less than $n^{1/3}$. 
Now consider the union of the lists for all of the six cones, which has size less than $6n^{1/3}$. 
Notice that it is the set of all points $p\in P$ with $r_p \geq |px|$, and thus the lemma holds. 
%
%
\end{proof}

By Lemma~\ref{chain:1}, we can compute all $L_p$'s in $O(n^{4/3}\log^4 n)$ time, and for each $L_p$, we can compute
six chains in $O(n^{1/3}\log n)$ time. Since the number of $L_p$'s is $O(n^{2/3})$, the total time for computing all chains of $\mathcal{C}$
is $O(n^{4/3}\log^4n)$ time. 

\subsection{Separation tree of \texorpdfstring{$\mathcal{R}$}{R}}\label{sec:remain}
In this section, we build a reachability oracle for $\mathcal R$, which is similar to the reachability oracle proposed by Kaplan et al. \cite[Section 4.2]{kaplan_et_al:LIPIcs:2015:5106}. 
In this case, since $\mathcal R$ is $O(n^{1/3})$-thick, we can derive a better result.
Then Lemma~\ref{4:separator} shows that there is a separator of size $O(n^{2/3})$. 
Recall that $R$ is the vertex set of $\mathcal{R}$. 

\paragraph{Data structure.}
We construct the separation tree $T$ of $R$ recursively as follows.
We compute a separator $S_\textnormal{cross}$ of $\mathcal R$ and two subsets $\sinn$ and $\sout$ separated by $S_\textnormal{cross}$. 
We recursively construct the separation trees of $\sinn$ and $\sout$.
Then we make a new node $v$, and connect $v$ with the roots of the separation trees of $\sinn$ and $\sout$. 
We let $G_v$ denote the subgraph of $G$ induced by $R$. 
%
See Figure~\ref{sep-1}.

\begin{figure}
	\centering
	\includegraphics[width=0.5\textwidth]{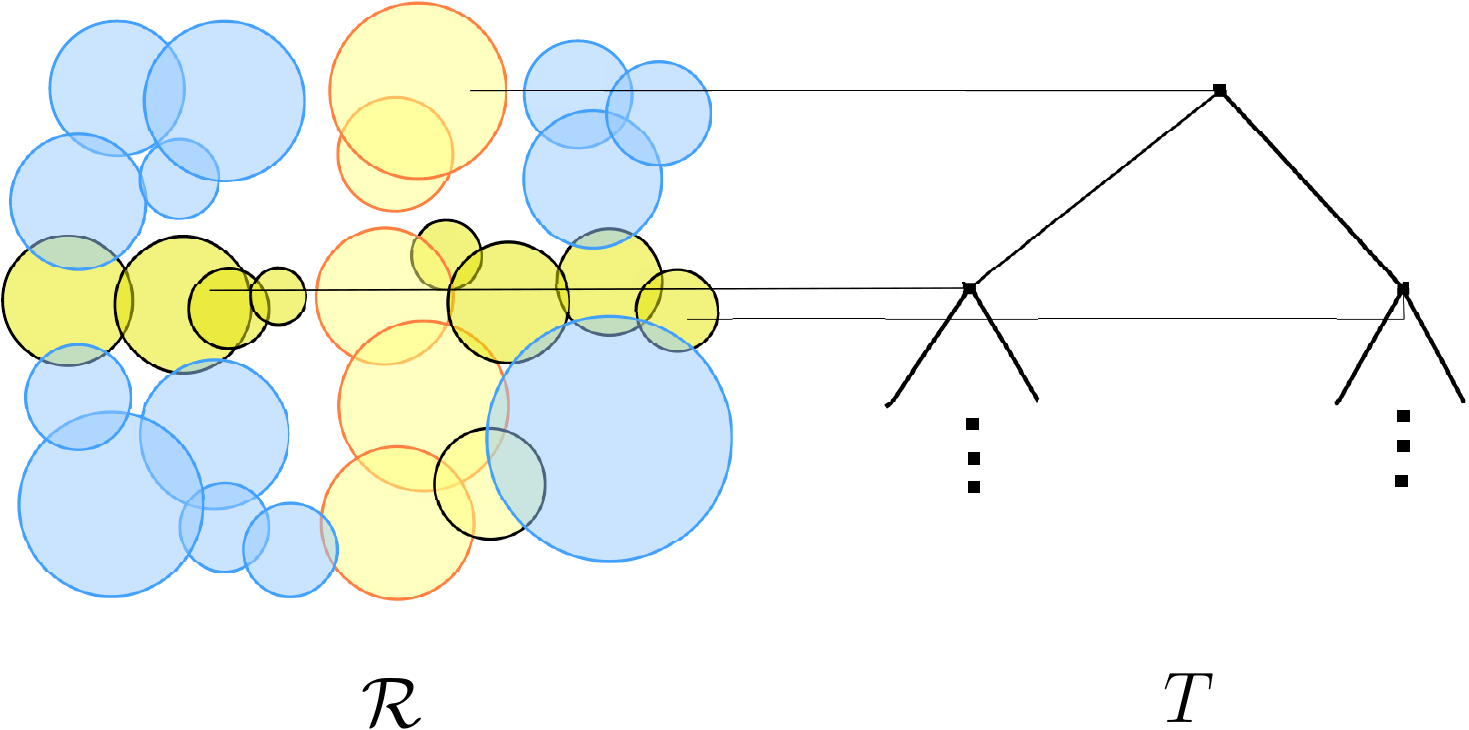}
	\hfil
	\caption{The remaining set and its separation tree}
	\label{sep-1}
\end{figure}

For each node $v$, we store the reachability information as follows: 
For every point $p\in G_v$, we store two lists of points of $S_\textnormal{cross}$ which is reachable to $p$ 
and which is reachable from $p$ within $G_v$. 
In particular, we construct a $2$-spanner of $G_v$. 
Then, for each point $s\in S_\textnormal{cross}$, we apply the BFS algorithm in Section~\ref{sec:bfs} from $s$. 
Also, we reverse the spanner and again apply the BFS algorithm from $s$.

\paragraph{Query Algorithm.}
Given two query points $p,q\in R$, we want to check if $q$ is reachable from $p$ in $\mathcal{R}$. 
To do this, we observe the following.
Let $v$ and $u$ be the two nodes of the separation tree $T$ such that the separators of $G_v$ and $G_u$ contain $p$ and $q$, respectively. 
They are uniquely defined because each point of $R$ is contained in exactly one separator stored in $T$. 
Let $L$ be the path of $T$ from the lowest common ancestor of $v$ and $u$ to the root. 
Consider a path $\pi$ from $p$ to $q$ in $\mathcal R$, if it exists. 
By construction, there is a node $w$ in $L$ such that the separator of $G_w$ 
intersects $\pi$. Among them, consider the node closest to the root node. 
Then $G_w$ contains $\pi$. 
Therefore, it suffices to check if $q$ is reachable from $p$ in $G_x$ for every node $x$ in $L$.

%
%

To use this observation, we first compute $v, u$ and $L$ in $O(\log n)$ time. Then
for each node $w$ of $L$, we check if there is a point $x$ in separator 
such that $p$ is reachable to $x$ and $q$ is reachable from $x$ in $O(m)$ time,
where $m$ denotes the size of the separator of $G_w$.
We return $\mathsf{YES}$ if and only if there is such a point $x$.
Since the size of the separators stored in each node is geometrically increasing along $L$, 
the total size is dominated by the size of the separator of $R$, which is $O(n^{2/3})$. 
Therefore, our query algorithm takes $O(n^{2/3})$ time.


\begin{lemma}\label{R:conclude}
We can construct a separation tree $T$ of $\mathcal R$ with associated reachability information in $O(n^{5/3})$ time and $O(n^{5/3})$ space. Then, we can query whether there is path from $p$ to $q$ in $\mathcal R$ within $O(n^{2/3})$ time.
\end{lemma}

\begin{proof}
	Since the analysis of the query time is presented in the above text, we focus on the size of the data structure and
	its preprocessing time only. 
	For each node $v$ of the separation tree, we spend $O(m)$ time to compute a separator and two separated subsets, $O(m \log^3 m)$ time to compute a 2-spanner, and $O(m^{2/3})\cdot O(m)$ = $O(m^{5/3})$ time for the BFS algorithm,
	where $m$ denotes the complexity of $G_v$. 
	
	Let $P(m)$ be the time for constructing the separation tree for a point set of size $m$.
	Then we have $P(m) \leq P(m_1)+ P(m_2) + O(m^{5/3})$,
	where $m_1$ and $m_2$ denote the size of $\sinn$ and $\sout$, respectively. 
	Notice that $m_1+m_2 \leq m$ and $m_1, m_2 < 2m/3$, and thus $P(n) = O(n^{5/3})$. 
	Similarly, we can show that the space complexity is $O(n^{5/3})$.
\end{proof}
%
%
%

\subsection{Chain Indices}\label{sec:chain}
In this section, we construct a reachability oracle for each chain $C\in\mathcal{C}$: 
Given any two points $p$ and $q$ in $P$, we can check if there is a path from $p$ to $q$ intersects $C$. 
For each chain $C=\langle p_1,...,p_t\rangle$, we can construct the oracle in $O(n)$ time once we have a 2-spanner of $G$. 
To do this, we need the following lemma. See Figure~\ref{index-1}.

\begin{lemma}\label{lem:chain-index}
For two points $p$ and $q$ in $P$, 
let $i$ be the largest index such that $p$ is reachable to $p_i$, and $j$ be the smallest index such that $p_j$ is reachable to  $q$. Then, there is a $p$-$q$ path that intersects $C$
if and only if $j\leq i$. 
\end{lemma} 

\begin{proof}
	If there is $p$-$q$ path that intersects $C$, 
	there exists $p_k\in C$ such that $p$ is reachable to $p_k$ and $p_k$ is reachable to $q$. Then, $j\leq k\leq i$.
	Conversely, if $j\leq i$, there exist a $p$-$p_i$ path, a $p_i$-$p_j$ path, and a $p_j$-$q$ path by the choice of $i$ and $j$. 
	Then, the concatenation of those three paths is a $p$-$q$ path that intersects $C$.
\end{proof}

\begin{figure}
\centering
   \includegraphics[width=0.6\textwidth]{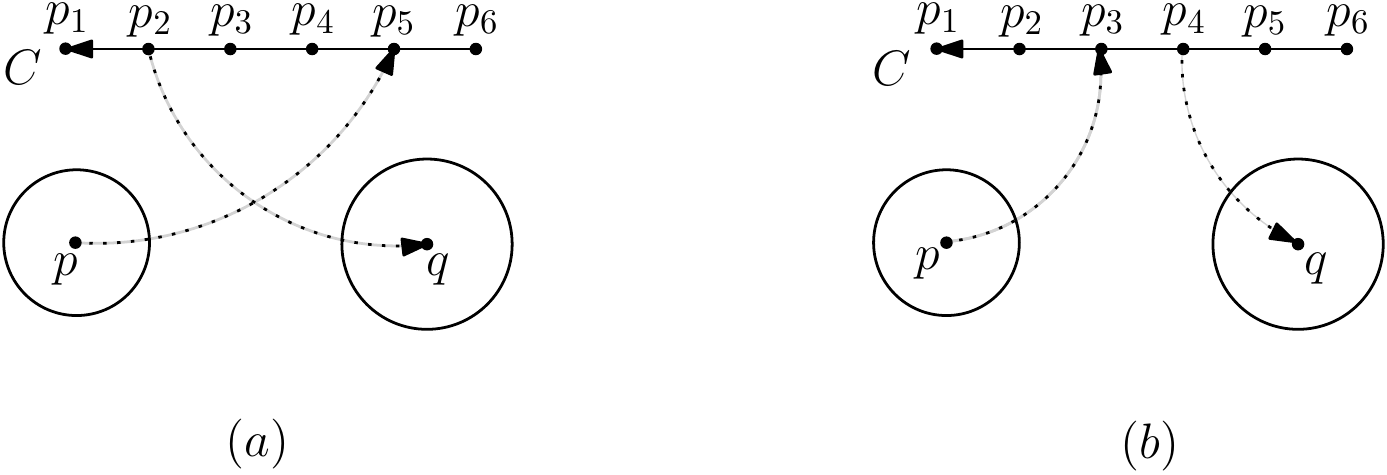}
   \hfil
\caption{(a) There is a $p$-$q$ path via $C$ if $j(q)\leq i(p)$.  (b) 
	There is no $p$-$q$ path that intersects $C$ if $j(q)>i(p)$}
\label{index-1}
\end{figure}

For every point $q\in P$,  
we store the largest index $i(q)$ such that $q$ is reachable to $p_{i(q)}$, and store the smallest index $j(q)$ 
such that $p_{j(q)}$ is reachable to  $q$. 

\begin{lemma}\label{lem:const-index}
	We can compute the indices $i(\cdot)$ and $j(\cdot)$ for every $q\in P$ and every $C\in\mathcal C$
	in $O(n^{5/3})$ time. Also, the total number of indices we store is $O(n^{5/3})$.
\end{lemma}

\begin{proof}
	In the following, we show how to compute $j(q)$ for every point $q\in P$. The other index $i(q)$ can be computed similarly.
	Let $I_j$ be the set of points $q\in P$ such that $j(q)=j$. 
	That is, a point $q$ is contained in $I_j$ if and only if $p_i$ is the first point of $C$ which is reachable to $q$. 

	We compute $I_1, \ldots, I_t$ one by one in order. 
	For an index $\ell$, assume that we maintain a graph $H_\ell$, which is the subgraph of the 2-spanner of $G$ 
	induced by $\cup_{j=\ell}^{t}I_j$. 
	We can compute $I_\ell$ by applying the BFS algorithm on $H_\ell$ starting from $p_{\ell}$.
	The set of all points we encountered is exactly $I_\ell$. 
	To maintain the invariant, we remove all points in $I_\ell$ and their adjacent 
	edges from $H_\ell$, and denote it by $H_{\ell+1}$. 
	For each index $\ell$, the BFS algorithm runs in $O(|I_\ell|)$ time. 
	Since $I_\ell$'s are pairwise disjoint, the BFS algorithm runs in $O(n)$ time in total. 
	Since we have $O(n^{2/3})$ chains, the total time for computing all chain indices is $O(n^{5/3})$. 
\end{proof}

\subsection{Reachability Oracles}

Given two points $p,q\in P$, we can check if $p$ is reachable from $q$ as follows.
Suppose there is a $p$-$q$ path $\pi$. 
If there is a chain $C$ that intersects $\pi$ at a point of $C$, say $p_k$. 
Then $j(q)\leq k\leq i(p)$ for the indices $j(q)$ and $i(p)$ stored in $C$
by Lemma~\ref{lem:chain-index}. 
In this case, we can find such a chain $C$ in $O(n^{2/3})$ time by computing indices $j(q)$ and $i(p)$ for all chains of $\mathcal C$. 
Otherwise, no chain of $\mathcal{C}$ intersects $\pi$.  
Then $\pi$ is contained in $\mathcal R$, and thus we can use the reachability oracle for $\mathcal R$ described in Section~\ref{sec:remain}.
This takes $O(n^{2/3})$ by Lemma~\ref{R:conclude}.  
%
%



\begin{theorem} \label{4:summary}
Given a set $P$ of points associated with radii, we can compute a reachability oracle for the transmission graph of $P$
in $O(n^{5/3})$ time. The reachability oracle has size $O(n^{5/3})$ and supports the query time $O(n^{2/3})$.
\end{theorem}

\section{Reachability Oracle for Small Radius Ratio}\label{sec:psioracle}
We improve the reachability of oracle described in Section~\ref{sec:oracle} if  $\Psi$ is polynomial in $n$.
In particular, we construct a reachability oracle with size $O(n^{3/2}\log^{1/2}\Psi)$, query time $O(n^{1/2} \log^{1/2} \Psi)$, and preprocessing time $O(n^{3/2}\log^{1/2}\Psi)$ time. 

\subsection{Hierarchical Grid}
For an index $i\in\{0,1,...\}$, consider the partition of the plane into axis-parallel squares (cells) with diameter $2^i$
such that the origin lies in the corner of a cell.
We call this the \emph{grid} at level $i$, and denote it by $\mathcal Q_i$. 
We consider the $L$ grids $\mathcal Q_0,...,\mathcal Q_L$ where $L$ = $\lceil \log{\Psi} \rceil$.
For each cell $\sigma \in \mathcal Q_i$, let $P_\sigma$ be the set of points $p$ in $P \cap \sigma$ such that $r_p \in [2^i,2^{i+1})$.
Note that every point is contained in $P_\sigma$ for exactly one cell $\sigma$ for all grids $\mathcal Q_1,\ldots, \mathcal Q_L$. 

We construct a new graph $H=(V',E')$ where $V'$ is the set of cells $\sigma$ with $P_\sigma \neq \emptyset$, and 
$E'$ is the set of pairs $(\sigma_1, \sigma_2)$ such that 
there are two points $p\in P_{\sigma_1}$ and $q\in P_{\sigma_2}$ with $\direct{p}{q}\in E$. 
Note that the points in $P_\sigma$ form a clique, and thus 
it suffices to construct a reachability oracle for $H$. 

\begin{lemma} \label{lem:H}
	We can construct $H=(V',E')$ from the transmission graph $G$ within $O(n\log n \log \Psi)$ time.
	Moreover, the number of edges of $H$ is $O(|V'|\log\Psi)$.
\end{lemma} 
\begin{proof}
	For every cell $\sigma$ with $P_\sigma \neq\emptyset$, we construct the power diagram of $P_\sigma$. 
	This takes $O(n\log n)$ time for all cells in total. 
	We define a $\emph{grid cluster}$ as a block of $9\times9$ contiguous grid cells.
	For every point $p\in P_\sigma$ and for every grid level $i=0,...,L$, 
	let $C(p,i)$ denote the grid cluster of grid level $i$ 
	whose center grid cell contains $p$. 
	
	For each index $i\in[1, L]$ and each cell $\sigma$, we want to find all edges $\direct{\sigma}{\sigma'}$ such that $\sigma'$ is a cell of $\mathcal Q_i$. 
	For a cell $\sigma'$ of $\mathcal Q_i$, $\direct{\sigma}{\sigma'}\in E'$ if and only if  
	$\direct{p}{q}\in E$ for two points $p\in \sigma$ and $q\in \sigma'$.
	Then, the Euclidean distance between $p$ and the center of $\sigma'$ is at most $2^i+2^{i+1}=3\cdot2^i$.
	Therefore, $\sigma'\in C(p,i)$.
	Thus, we can compute 
	all edges $\direct{\sigma}{\sigma'}$ such that $\sigma'$ is a cell of $\mathcal Q_i$ 
	by considering the cells in $C(p,i)$ for every $p\in P_\sigma$. 
	In particular, for every cell $\sigma'$ in $C(p,i)$, 
	we check if there is a point $q\in P_{\sigma'}$ with $\direct{q}{p}\in E$ using the power diagram 
	within $O(|P_\sigma|\log n)$ time.
	Since $C(p,i)$ has $O(1)$ cells, we can do this for all cells in $C(p,i)$ in  $O(|P_\sigma|\log n )$ time.
	In total, we can compute every edge in $E'$ in $O(n\log n \log \Psi)$ time since $L=\log\Psi$ and the total size
	of $P_\sigma$ for all levels $i$ and all cells in $\mathcal Q_i$ is $n$.

	We want to compute the number of edges in $E'$. To do this, we first consider
	$(\sigma, \sigma')\in E'$ such that $\sigma\in \mathcal Q_i$ and $\sigma\in \mathcal Q_j$ with $j<i$.
	Then, there are two points $p\in P_\sigma$ and $q\in P_{\sigma'}$ such that $\direct{p}{q}\in E$.
	Since $r_q>r_p$ by construction, $\direct{q}{p}\in E$, and thus $(\sigma',\sigma)\in E$.
	Therefore, it suffices to compute the number of edges in $E'$ from $\sigma'\in \mathcal Q_j$ to $\sigma\in \mathcal Q_i$ where $i\leq j$.

	For a cell $\sigma\in \mathcal Q_i$, we show that the number of incoming edges from cells in $\mathcal Q_j$ to $\sigma$ is $O(1)$ for every $j\geq i$.
	Let $\sigma_j$ be a cell in $\mathcal Q_j$ such that $\sigma\subseteq \sigma_j$.
	Then, for every $p\in P_\sigma$, $C(p,i)$ is equal to the grid cluster with center cell $\sigma_j$. 
	Therefore, all incoming edges for $\sigma$ from the cells in $\mathcal Q_j$ are contained in the grid cluster with center cell $\sigma_j$.
	Then, the total number of incoming edges from the cells in $\cup_{k=i}^L \mathcal Q_k$ to $\sigma$ is $O(\log\Psi)$. 
	This shows that the number of edges of $H$ is $O(|V'|\log\Psi)$.
\end{proof}

\subsection{Separation tree revisited}
Now we want construct a reachability oracle for $H$. 
We construct the separation tree described in Section \ref{sec:remain}.
We slightly make the additional step because now the graph $H$ consists of cells instead of points with associated radii.
For each cell $\sigma\in V'\cap \mathcal Q_i$ with its center $c(\sigma)$, we associate radius $3\cdot2^i$ with $c(\sigma)$ and let $D(\sigma)$ be the associated disk.
We denote the set of all centers of cells of $V'$ by $X$.

We first show that $X$ is $O(\log\Psi)$-thick. That is, any point $p$ in the plane is contained in $O(\log\Psi)$ associated disks of 
the points of $X$.  
For a grid level $i$ and a cell $\sigma \in \mathcal Q_i$, consider the associated disk of $c(\sigma)$ containing $p$.
Then, $\sigma$ is contained in
 the block of $9\times9$ contiguous grid cells 
of level $i$ whose center grid cell contains $p$. 
Therefore, $p$ is contained in $O(1)$ associated disks of $c(\sigma)$ for cells of grid level $i$. 
Since the number of grids is $L=O(\log\Psi)$, $X$ is $O(\log\Psi)$-thick.

\paragraph{Data structure.}
We construct the separation tree $T$ for $X$ as follows.
We compute a separator $S_\textnormal{cross}$ of $X$ and two subsets $\sinn$ and $\sout$ separated by $S_\textnormal{cross}$. 
We recursively construct the separation trees of $\sinn$ and $\sout$.
Then we make a new node $v$, and connect $v$ with the roots of the separation trees of $\sinn$ and $\sout$. 
Let $H_v$ be the subgraph of $H$ induced by $X$. 

For each node $v$ of $T$, we store the reachability information as follows.
For every point $c(\sigma)\in S_\textnormal{cross}$, we store two lists of cells of $H_v$ which are reachable to $\sigma$ 
and which are reachable from $\sigma$ within $H_v$.
In particular, for each cell $\sigma$ where $c(\sigma)\in S_\textnormal{cross}$, we apply a breadth-first search from $\sigma$ in $H_v$. 
Also, we reverse $H_v$ and again apply  a breadth-first search from $\sigma$. 

\paragraph{Construction time.}
For each node $v$ of $T$, we spend $O(m)$ time to compute a separator and two separated subsets, where $m$ denotes the vertices of
 $H_v$.  The size of the separator is $O(m^{1/2}\log^{1/2}\Psi)$ because $X$ is $O(\log\Psi)$-thick.
Moreover, the number of edges of $H_v$ is $O(m\log\Psi)$ by Lemma~\ref{lem:H}.  
Thus we can apply a breath-first search in $O(m^{3/2}\log^{3/2}\Psi)$ time. 

Let $P(m)$ be the time for constructing the separation tree for a point set of size $m$.
Then we have $P(m) \leq P(m_1)+ P(m_2) + O(m^{3/2}\log^{3/2}\Psi)$,
where $m_1$ and $m_2$ denote the size of $\sinn$ and $\sout$, respectively. 
Notice that $m_1+m_2 \leq m$ and $m_1, m_2 < 2m/3$, and thus $P(n) = O(m^{3/2}\log^{3/2}\Psi)$. 
Similarly, 
we can show that the space complexity is $O(m^{3/2}\log^{1/2}\Psi)$.
But in this case, the space used by each node of $T$ is 
$O(n^{3/2}\log^{1/2}\Psi)$  instead of $O(n^{3/2}\log^{3/2}\Psi)$.

\paragraph{Query Algorithm.}
Given two query cells $\sigma_1,\sigma_2 \in H$, we want to check if $\sigma_2$ is reachable from $\sigma_1$ in $H$. 
To do this, we observe the following.
Let $v$ and $u$ be the two nodes of the separation tree $T$ such that the separators of $G_v$ and $G_u$       contain $c(\sigma_1)$ and $c(\sigma_2)$, respectively. 
Let $L$ be the path of $T$ from the lowest common ancestor of $v$ and $u$ to the root. 


Note that $D(\tau)$ intersects $D(\tau')$ for any two cells $\tau$ and $\tau'$ of $V'$ with $\direct{\tau}{\tau'}\in E'$.  
Consider a node $w\in T$ and the separator $S_\textnormal{cross}$ of $G_w$. Let $\sinn$ and $\sout$ be the two sets separated by $S_\textnormal{cross}$. 
Also, for any two cells $\sigma$ and $\sigma'$ of $V'$ with $c(\sigma)\in \sinn$ and $c(\sigma')\in \sout$, 
a $D(\sigma)$-$D(\sigma')$ path in the (undirected) disk intersection graph of $X$ intersects the associated disk of  a point of $S_\textnormal{cross}$.
Therefore, a $\sigma$-$\sigma'$ path in $H$ intersects a cell $\sigma''$ where $c(\sigma'')\in S_\textnormal{cross}$.

This implies that for any path $\pi$ from $\sigma_1$ to $\sigma_2$ in $H$, 
there is a node $x$ in $L$ such that 
$\pi$ intersects a cell $\sigma$ such that $c(\sigma)\in S$, where $S\subseteq X$ denotes the separator of $G_x$. 
Among them, consider the node closest to the root node. 
Then $G_w$ contains $\pi$. 
Therefore, it suffices to check if $\sigma_2$ is reachable from $\sigma_1$ in $G_w$ for all nodes $w\in L$. 

%
%

To use this observation, we first compute $v, u$ and $L$ in $O(\log n)$ time since $T$ has $O(\log n)$ levels. 
Then for each node $u$ of $L$, we check if the separator of $G_u$ contains 
$c(\sigma)$ such that $\sigma_1$ is reachable to $\sigma$ and $\sigma_2$ is reachable from $\sigma$ in $O(k)$ time,
where $k$ denotes the size of the separator of $G_u$.
We return $\mathsf{YES}$ if there is such a node $u$.
Otherwise, we return $\mathsf{NO}$. 
Since the size of the separators stored in each node is geometrically increasing along $L$, 
the total size is dominated by the size of the separator of $X$, which is $O(n^{1/2}\log^{1/2}\Psi)$. 
Therefore, our query algorithm takes $O(n^{1/2}\log^{1/2}\Psi)$ time.





\begin{theorem} \label{Psi:conclude}
Given a set $P$ of points associated with radii and $P$ has radius ratio $\Psi$, we can compute a reachability oracle 
for the transmission graph of $P$ in $O(n^{3/2}\log^{3/2}\Psi)$ time. The reachability oracle has size $O(n^{3/2}\log^{1/2}\Psi)$ and supports the query time $O(n^{1/2}\log^{1/2}\Psi)$.
\end{theorem}

\section{Continuous Reachability Oracle}\label{sec:geom}
%
In this section, we present a continuous reachability oracle which its complexity is independent of the radius ratio $\Psi$. In particular, our data structure has size $O(n^{5/3})$ so that for any two point $s\in P$ and $t\in\mathbb{R}^2$, we can check if $s$ is reachable to $t$ in $O(n^{2/3}\log^2 n)$ time. Also, this data structure can be constructed in $O(n^{5/3})$ time. 
If $t$ is reachable from $s$, there is a point $p\in P$ reachable from $s$ with $t\in D_p$. 
In this case, we define a \emph{$s$-$t$ path} in $G$ as the concatenation of a $s$-$p$ path in $G$ and the segment connecting $p$ and $t$.

\medskip 
Consider two query points $s\in P$ and $t\in\mathbb{R}^2$.
If there is a $s$-$t$ path $\pi$, 
we denote the vertex incident to $t$ in $\pi$ by $p(\pi)$. 
We construct auxiliary data structures for $R$ and $\mathcal{C}$ to handle the following two cases. 
We first consider the case that there is a $s$-$t$ path $\pi$ with $p(\pi)\in R$. 
In this case, 
we choose a set $R_0$ of $O(1)$ points in $R$ so that 
there is a $s$-$t$ path $\pi'$ with $p(\pi')\in R_0$ if and only if there is a $s$-$t$ path $\pi$ with $p(\pi)\in R$. 
If it is not the case, for any $s$-$t$ path $\pi$, $p(\pi)$ is contained in a chain of $\mathcal C$. 
We can handle this by investigating every chain of $\mathcal{C}$, and finding the first point in the chain whose associated disk contains $t$. 
In addition to this, we construct the discrete reachability oracle for $G$ described in Section~\ref{sec:oracle}.

\subsection{The remaining set $R$, revisited} \label{sec:r}


We construct the data structure so that we can check if there is a $s$-$t$ path $\pi$ with $p(\pi)\in R$. 
To do this, wee construct the $O(n)$-sized data structure proposed by Afshani and Chan~\cite{Chan:SIAM2009} such that for $P$ and a query point $t$, we can find all points in $P$ whose associated disks contain $t\in\mathbb{R}^2$ in $O(\log n +k)$ time, where $k$ is the number of disks that contain $t$. Moreover, it can be constructed in $O(n\log n)$ time.
Since $R$ is $O(n^{1/3})$-thick, this query time is bounded by $O(n^{1/3})$.  



Given two points $p\in P$ and $t\in\mathbb{R}^2$, we compute a set $R_t$ of $O(n^{1/3})$ points of $P$ whose associated disks contain $t$ within $O(n^{1/3})$ time. 
Then we choose a subset $R_0$ of $R_t$ of size $O(1)$ such that 
there is a $s$-$t$ path $\pi'$ with $p(\pi')\in R_0$ if and only if there is a $s$-$t$ path $\pi$ with $p(\pi)\in R$. 

\begin{lemma} \label{G:R}
Assume that we are given a point $t\in\mathbb{R}^2$ and a set $R_t$ of points of $R$ whose associated disks contain $t$. 
We can compute a $O(1)$-sized set $R_0\subset R_t$ such that 
$R_0 \cap D_p \neq \emptyset$ for every point $p\in R_t$ in $O(|R_t|)$ time. 
\end{lemma}

\begin{proof}
	We consider six interior-disjoint cones with opening angle $\pi/3$ with apex $t$.
	For a cone $F$, we consider the list $L$ of the points $p\in R_t\cap F$.
	We pick the point $q\in L$ that minimizes the distance value $d_F(t,\cdot)$. (For the definition of $d_F(t,\cdot)$, see the beginning of Section~\ref{sec:theta}.) 
	We claim that for every point $x\in L$, the associated disk of $x$ contains $q$.
	Let $\triangle$ be the regular triangle surrounded by the two rays of $F$ and a line passing through $x$. 
	Then, $t$ is one of the farthest points from $x$ within $\triangle$. Also, $q$ is contained in $\triangle$ by the definition of $q$.
	Therefore, $|qx|\leq |tx|\leq r_x$, and thus $D_x$ contains $q$.
	Then, the set of points $q$ from all cone $F$ satisfy the condition of $R_0$. 
	This procedure takes $O(n^{1/3})$ time.
\end{proof}

Then we can answer the continuous reachability query using the discrete reachability oracle for all points $q\in R_0$ 
in $O(n^{2/3})$ time by Theorem~\ref{4:summary}.

\subsection{The set $\mathcal C$ of chains, revisited.}\label{sec:c}
We construct a data structure for each chain $C\in \mathcal C$ so that we can compute the first point in $C$ which contains $t$.
To do this, we construct a balanced binary search tree of the indices in $[1,t]$ for $C=\langle p_1,\ldots,p_t \rangle$.
For each node $u$ of the binary search tree, we construct the power diagram of the points stored in the subtree rooted at $u$.
Note that this data structure is a variation of the third-level tree proposed in Section \ref{sec:range}.
We sort the points along their indices here, while we sort the points along $\ell$-projections in Section \ref{sec:range}.
Therefore, as we showed in Section \ref{sec:range}, the construction takes $O(m\log m)$ time, and we can compute the first point in $C$  which contains $t$ within $O(\log^2 m)$ time for each chain $C$, where $m=|C|$.
In this way, we can construct the auxiliary data structures for all chains of $\mathcal C$ in $O(n\log n)$ time.
Given two points $s\in P$ and $t\in\mathbb{R}^2$, we can check if $s$ is reachable to $t$ as follows.
Suppose there is a $s$-$t$ path $\pi$. 
If $p(\pi)$ is contained in a chain $C\in \mathcal C$, let $k$ be the index of $p(\pi)$ in $C=\langle p_1,\ldots,p_t\rangle$,
that is, $p_k=p(\pi)$. 
We let $j(t)$ denote the index of the first point in $C$ which contains $t$, and let $i(s)$ denote the index of the last point in $C$
which is reachable from $s$. Recall that $i(s)$ is stored in the discrete reachability oracle, and $j(t)$ can be computed
using the auxiliary data structure for $C$ as mentioned above. 
Then there is a $s$-$t$ path $\pi$ with $p(\pi)\in C$ if and only if  $j(t)\leq k\leq i(s)$. 
We do this for all chains in $\mathcal{C}$.
Since we can compute the first point that contains $t$ for every chain of $\mathcal C$ in $O(n^{2/3}\log^2 n)$ time,
the total query time is $O(n^{2/3}\log^2 n)$ time. 
Therefore, we have the following theorem.

 \begin{theorem} \label{5:conclude}
Given a set $P$ of points associated with radii, we can compute a 
continuous reachability oracle for the transmission graph of $P$ in $O(n^{5/3})$ time. 
The reachability oracle has size $O(n^{5/3})$ and supports the query time $O(n^{2/3}\log^2 n)$.
\end{theorem}

\end{document}